\documentclass[11pt]{article}
\usepackage[letterpaper,margin=1in]{geometry}
\usepackage{setspace} % double space
\usepackage[T1]{fontenc}
\usepackage{amsthm}
\usepackage{amsmath}
\usepackage{array}
\usepackage{amsfonts}
\usepackage{graphicx}
\usepackage{subfigure}
\usepackage{dcolumn}
\usepackage{mathrsfs}
\usepackage{float}
\usepackage{multirow}
\usepackage{amsmath}
\usepackage{amssymb}
\usepackage{braket}
\usepackage{color}
\usepackage{tcolorbox}
\usepackage[colorlinks,citecolor=blue,linkcolor=blue]{hyperref}
\usepackage[ruled]{algorithm2e}
\usepackage{tikz}
\usepackage{braket}
\usepackage{appendix}
\usepackage{circuitikz}
\usepackage{framed}
\usepackage{tcolorbox}
\usepackage{booktabs}
\usepackage{qcircuit}
\newtheorem{lemma}{Lemma}

\newtheorem{theorem}{Theorem}
\begin{document}
\title{Quantum algorithm for learning  secret strings and its experimental demonstration}
\author{  Yongzhen Xu\thanks{These authors contributed equally to this work.}, Shihao Zhang\footnotemark[1] \thanks{ zhangshh63@mail.sysu.edu.cn (S. Zhang), lilvzh@mail.sysu.edu.cn (L. Li) }  ~and Lvzhou Li \footnotemark[2] \\
\small{{\it Institute of Quantum Computing and Computer Theory,}}\\
\small{{\it School of Computer Science and Engineering,}}\\
\small {{\it  Sun Yat-sen University, Guangzhou 510006, China}}}
\date{\today }

\maketitle
\begin{abstract}
In this paper, we consider the   secret-string-learning  problem in the teacher-student setting: the teacher has a secret  string $s\in {{\{0,1\}}^{n}}$, and the student wants to learn the secret $s$ by question-answer interactions with the teacher, where  at each time, the student can ask the teacher with  a pair
 $(x, q) \in {{\{0,1\}}^{n}}\times\{0,1,\cdots, n-1\}$ and the teacher returns a bit given by  the oracle $f_{s}(x,q)$ that indicates  whether the length of the longest common prefix  of $s$ and $x$ is greater than $q$ or not.  Our contributions are as follows.
 \begin{itemize}
  \item[(i)] We prove that  any classical deterministic algorithm needs at least $n$ queries   to the oracle $f_{s}$ to learn the $n$-bit secret string $s$  in both the worst case and  the average case, and also present an optimal classical deterministic algorithm  learning any  $s$ using $n$ queries.
  \item[(ii)] We obtain a quantum algorithm  learning the $n$-bit secret string $s$ with certainty using $\left\lceil n/2\right\rceil$ queries to the  oracle $f_s$, thus proving a double speedup over classical counterparts.
  \item[(iii)] Experimental demonstrations of our quantum algorithm on the IBM cloud quantum computer  are presented, with average success probabilities of $85.3\%$ and  $82.5\%$ for  all cases with $n=2$ and $n=3$ , respectively.
\end{itemize}

\end{abstract}

\section{Introduction}

Strings are one of the most basic structures in mathematics and computer science.
A string is an abstract data structure consisting of a sequence of zero or more letters over a non-empty finite alphabet. The study of string processing methods is a fundamental concern in computer science, which turns out to have a wide range of applications in areas such as information theory, artificial intelligence, computational biology and linguistics.
String problems have been extensively studied in classical computing for decades, and numerous effective algorithms to deal with them have been proposed, including exact string  matching \cite{Knuth1977,Boyer1977},
finding patterns in a string \cite{Manacher1975,Apostolico1995}, and many others ~\cite{Crochemore2007}.

Naturally, string problems have also attracted much attention from the quantum computing community over the past two decades. In 2003, a quantum algorithm for exact string matching by Ramesh and Vinay provided a near-quadratic speedup over the fastest known classical algorithms~\cite{Hariharan2003}. In 2017 Montarano developed a  quantum algorithm that is super-polynomially faster than the best possible classical ones for $d$-dimensional pattern matching problem on average-case inputs~\cite{Montanaro2017}. In 2021, Niroula and Nam designed a quantum pattern matching algorithm from the perspective of circuit implementation \cite{Niroula2021}.
Recently, some other string problems have also been investigated by focusing on novel quantum algorithms \cite{Boroujeni2021,  Gall2022,  Akmal2022}.

As one of the most important string  problems, string learning  has attracted much interest in  both the classical and quantum computing  models since it has interesting applications in data mining and cyber  security. Typically, there are two parties called a teacher and a student. The teacher has a secret bit string $s$ of length $n$ and the student wants to identify this secret string
by asking  a certain number of queries to an oracle that answers some piece of information of $s$. The goal is to learn $s$ using as few queries as possible.
To date, there have been many wonderful results in quantum and classical settings for this problem with  different oracle types,  including  the index oracle \cite{Dam1998}, inner product oracle \cite{BV1997},
substring oracle \cite{Skiena1995,Cleve2012}, balance oracle in counterfeit coin problem\cite{Iwama2012}, and subset OR oracle in combinatorial group testing \cite{Du2000,Ambainis2014,Belovs2015}.

In this paper, we explore  the power of another oracle  as the length of the Longest Common Prefix (LCP) of two strings,  which is a well-known string similarity metric used in  data structures and algorithms \cite{Kasai2001,Bonizzoni2021}. Afshani et al. \cite{Afshani2009} used this oracle to consider the hidden permutation problem that comes from Mastermind game and evolutionary computation. To the best of our knowledge, there has not been any work related to string learning based on this oracle in the  quantum setting.  Hence, in this paper we consider quantum algorithms based on the LCP information to learn a secret string. More specifically,  the problem  is described  in the teacher-student setting: the teacher has a secret bit  string $s\in {{\{0,1\}}^{n}}$, and the student wants to learn the secret $s$  using as few queries to the teacher as possible, where  at each time, the student can ask the teacher with  a pair
 $(x, q) \in {{\{0,1\}}^{n}}\times\{0,1,\cdots, n-1\}$ and the teacher returns a bit given by  the oracle $f_{s}(x,q)$ that indicates  whether the length of the LCP  of $s$ and $x$ is greater than the integer $q$ or not. We manage to present a quantum algorithm that can offer an advantage over the best classical algorithm for solving this problem.

The main results of this paper are stated as follows. First, we prove that   any classical deterministic algorithm needs at least $n$ queries to
$f_{s}$ to learn the $n$-bit secret string $s$  in both the worst case and  the average case, and also  present a classical deterministic algorithm learning any $s$  with exact $n$ queries as an optimal one.
Second, we propose a quantum algorithm that learns  the $n$-bit secret string $s$ with certainty using $\left\lceil n/2\right\rceil$ queries to the  oracle $f_s$, thus proving a double speedup over classical counterparts.
Third, we demonstrate our quantum algorithm on the IBM cloud quantum computer using quantum circuit synthesis, compilation and optimization techniques. In the experiment, the average success probabilities for all cases with $n=2$ and $n=3$ achieve $85.3\%$ and $82.5\%$, respectively, which show the feasibility of implementing our quantum algorithm. Finally, we conclude this paper and put forward an interesting open problem that is worthy of further study.

\section{Results}\label{Results}

\subsection{The problem: learning a secret string}
We can describe the problem (learning a secret string) in the teacher-student setting: the teacher has a $n$-bit secret  string $s\in {{\{0,1\}}^{n}}$, and the student wants to learn the secret $s$ by question-answer interactions with the teacher. In this paper, we suppose that at each time, the student can ask the teacher with  a pair
 $(x, q) \in {{\{0,1\}}^{n}}\times\{0,1,\cdots, n-1\}$, and the teacher returns a bit given by the oracle
  \begin{equation}\label{fsqx}
    f_{s}(x,q) :=
    \begin{cases}
     0, & lcp(s,x)\le q;\\
       1, & lcp(s,x)>q,
    \end{cases}
\end{equation}
where
 \begin{equation}\label{Ksx}
  lcp(s,x):=\max \{i\in \{0,1,\cdots,n\}| \forall j\le i:{{x}_{j}}={{s}_{j}}\}
\end{equation}
 represents the length of the Longest Common Prefix (LCP) of $s$ and $x$.

When designing algorithms for solving this problem,
we hope to query the oracle ${{f}_{s}}$  as little as possible.

\subsection{Optimal classical algorithm}
  In this section, we propose a classical deterministic  algorithm for the secret string learning problem defined above, and further prove its optimality in terms of both the worst-case and the average-case query complexity.

\begin{algorithm}[htp]\label{alg1}
    \caption{Classical algorithm for learning secret string $s\in {{\{0,1\}}^{n}}$}
    \label{algorithm:test}
    \LinesNumbered
    \KwIn { A query oracle $f_{s}$ defined in Eq.~\eqref{fsqx}.}
    \KwOut {The secret string $s$.}
    $x\leftarrow0^n$.

    \For{$q = 0$ to $n-1$}{
        \If{$f_s(x,q) =0 $} {
            $x_{q+1} \leftarrow  x_{q+1}\oplus 1 $ ;

        }

    }
    \textbf{return} $s \leftarrow x$;
\end{algorithm}

\begin{theorem}\label{theorem1}
(1) There is a classical deterministic algorithm  learning any $n$-bit secret string $s$ using $n$ queries to the oracle $f_{s}$; (2) Any classical deterministic algorithm needs at least $n$ queries  to learn the secret string $s$  in both the worst case and  the average case.
\end{theorem}

\begin{proof}
For solving the problem, we propose a classical deterministic  algorithm  named Algorithm~\ref{alg1}, which starts from querying $(x=0^n,q=0)$ and then assigns the query data in each step  depending on the results from previous queries by Eq.\eqref{fsqx}. In this way,  each time one bit of the secret string $s$ is identified using one query according to the previous query outcomes, and as a result the total query complexity is exactly $n$ for any secret $s$.  More intuitively, the whole process of Algorithm~\ref{alg1}  can be described as a binary decision tree with height $n$ indicating its query complexity, and we present the case with  $n=3$  in Figure~\ref{fig1} for illustration.

Now  we  prove the lower bound for the query complexity of this problem (i.e. in the worst case) by information-theoretic argument. In general, any classical deterministic  algorithm for this problem can be described in terms of a binary decision tree that queries different $(x,q)$ in each internal node and identifies secret $s$ in each leaf node (i.e. external node), with its height being the query complexity of the algorithm.
Considering the fact that any binary tree with height $h$ has at most $2^h$ leaf nodes, we conclude that the height of any binary tree with total $2^n$ leaf nodes is at least $n$. That is,
  any classical deterministic   algorithm for the secret learning problem needs at least $n$ queries in the worst case.

 Moreover, we can prove the average-case query complexity of any classical deterministic  algorithm is no less than $n$.
 Note that the path length from the root node to a leaf node in a binary tree  indicates the number of queries to identify a secret string $s$ in an algorithm,
 our proof can be derived from  a fact about binary trees revealed as follows.
The external path length (EPL) of a tree is defined as the sum of path lengths of all its leaf nodes, and the minimum value of the EPL of a binary decision tree  with $N$ leaves is known as
$N(\log_2 N  +1+\gamma -2^{\gamma})$ with $\gamma= \left\lceil \log_2 N\right\rceil - \log_2 N \in [0,1)$ \cite{Knuth1998}. Hence, the minimum average path length of all $N$ leaf nodes is $\log_2 N  +1+\gamma -2^{\gamma}$. Here in our case we consider binary trees corresponding to classical query algorithms with $N=2^n$ and $\gamma=0$, where the minimum EPL is $n2^n$ and  the minimum average-case query complexity is $n$.

 In summary, our classical Algorithm~\ref{alg1} achieves optimality in terms of both worst- and average-case query complexity.
\end{proof}

\begin{figure}[htp]
    \centering
    \includegraphics[width=0.85\textwidth]{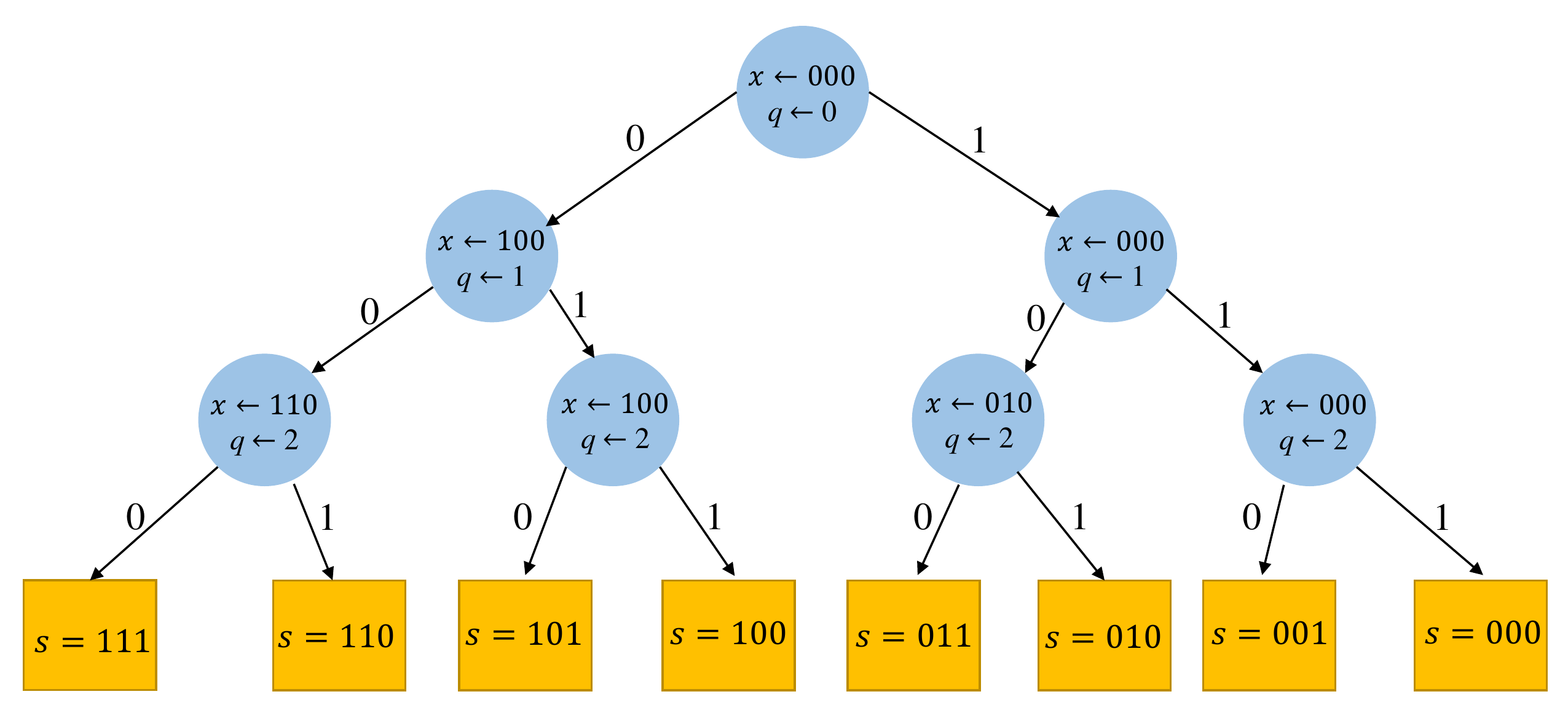}
    \caption{Binary decision tree for illustrating classical  Algorithm~\ref{alg1} with  $n=3$. The value 0 or 1 along each edge is the outcome after querying $(x,q)$ in blue circle nodes under a certain secret  $s\in {\{0,1\}^{3}} $ by Eq.~\eqref{fsqx}, and any secret string $s$ can be identified in an orange square node (i.e. leaf node) after three queries.}
    \label{fig1}
\end{figure}

\subsection{Quantum algorithm with speedup}

In this section, we present a quantum algorithm  that shows a speedup over any classical deterministic algorithm for learning the secret string $s$.

The  quantum oracle associated with  $f_s(x,q)$ in Eq.~\eqref{fsqx} is a quantum unitary operator, ${{O}_{s}}$, defined by its action on the
computational basis:
\begin{equation}\label{qxy}\left| x,q \right\rangle \left| y \right\rangle \xrightarrow{{{O}_{s}}}\left| x,q \right\rangle \left| y\oplus {{f}_{s}}(x,q) \right\rangle,
\end{equation}
where $\left| x,q \right\rangle$ is the query register and $\left| y \right\rangle$ is a single oracle qubit. Therefore, when we initialize the oracle qubit as $\left| y \right\rangle =\left| - \right\rangle = (\left| 0 \right\rangle -\left| 1 \right\rangle )/\sqrt{2} $, we have
\begin{equation}
 \left| x,q \right\rangle \left| - \right\rangle \xrightarrow{{{O}_{s}}}
    \begin{cases}
     \left| x,q \right\rangle \left| - \right\rangle, & {f_s}(x,q)=0;\\
       -\left| x,q \right\rangle \left| - \right\rangle, & {f_s}(x,q)=1,
    \end{cases}
\end{equation}
which can be summarized as
\begin{equation}
\left| x,q \right\rangle \left| - \right\rangle \xrightarrow{{{O}_{s}}}{{(-1)}^{{{f}_{s}}(x,q)}}\left| x,q \right\rangle \left| - \right\rangle.
\end{equation}
Note the state of the oracle qubit remains unchanged in this process, and thus can be omitted in the description of the action of  oracle $O_s$ for convenience as
\begin{equation}\label{qxOs}
   \left| x,q \right\rangle \xrightarrow{{O}_{s}}{{(-1)}^{{{f}_{s}}(x,q)}}\left| x,q \right\rangle.
\end{equation}

In fact, similar ¡°phase kickback¡± effects \cite{Cleve1998} and associated oracle simplifications have been explored and exploited in adapting some well-known quantum query algorithms, including Deutsch-Jozsa algorith \cite{Deutsch1993,Collins1998}, Bernstein-Vazirani algorithm \cite{Bernstein1997,Du2001}, and Grover search algorithm \cite{Grover1996,Figgatt2007}. In the following, we can directly design quantum query algorithms using the quantum oracle $O_s$ in Eq.~\eqref{qxOs} when all other involved operators  only act on the query register.

\begin{figure}[htp]
    \centering
    \includegraphics[width=0.95\textwidth]{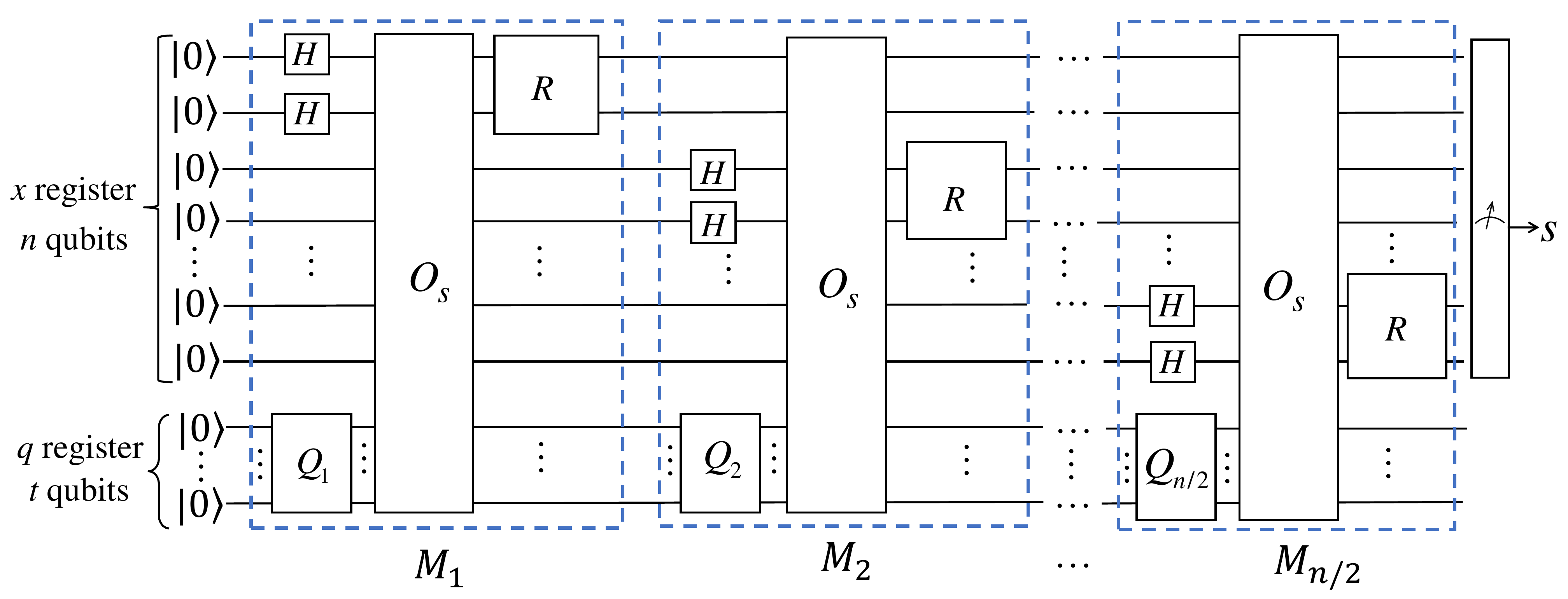}
    \caption{Schematic of the overall circuit for quantum learning algorithm with  even $n$, which applies $n/2$ subroutines $M_i$ ($i$=1,2,...,n/2) to the $n$-qubit $x$ register together with the $t$-qubit $q$ register initialized as $\left| {0}^{(n+t)} \right\rangle$
    . The working principle of each $M_i$ that consists of a quantum oracle $O_s$ in Eq.~\eqref{qxOs} and several fixed operators is illustrated in Lemma \ref{Mi}, and  the output of the final measurement is exactly the target secret string $s$.}
    \label{figQuanAL}
\end{figure}

\begin{theorem}\label{theorem2}
There is a quantum algorithm  learning the $n$-bit secret string $s$ with certainty using $\left\lceil n/2\right\rceil$ queries to the  oracle $f_s$.

\end{theorem}

\begin{proof}
 We   describe our quantum  algorithm for addressing  the case with even $n$ as depicted in Figure~\ref{figQuanAL}, and the odd $n$ cases can be handled in a similar way. The essence of our algorithm is to identify two bits of any $n$-bit secret string $s$ at a time by applying one quantum oracle  combined with other operators to a certain quantum superposition state, and thus the total query number is $n/2$ throughout the algorithm.

For even $n$ and  $t=\left\lceil {\log _2}(n) \right\rceil $, the input data $(x,q)$ is encoded by  a $x$ register with $n$ qubits combined with a $q$ register with $t$ qubits, and the whole quantum circuit  consists of $n/2$ subroutines $\{{M_i:i=1,2,...,n/2}\}$, where each $M_i$ can be broken into  four steps:

(1) Apply two Hadamard gates ${H}^{\otimes 2}$ to qubits $2i-1$ and $2i$ in the $x$ register.

(2) Apply an operator
$Q_i$
to the $q$ register
that can transform a $t$-qubit state  $\left| {q}^{(i-1)} \right\rangle$ into $\left| {{q}}^{(i)} \right\rangle$ with $q^{(0)}=0$ and $q^{(i)}=2i-1$ for $i\ge1$.

(3) Apply the quantum oracle operator $O_{s}$ described in 	Eq.~\eqref{qxOs} to all $(n+t)$ qubits.

(4) Apply a $4\times 4$ unitary operator
\begin{equation}\label{Rmatrix}
   R=\frac{1}{2} \left( \begin{matrix}
   -1 & 1 & 1 & 1  \\
   1 & -1 & 1 & 1  \\
   1 & 1 & -1 & 1  \\
   1 & 1 & 1 & -1  \\
\end{matrix} \right)
\end{equation} to qubits $2i-1$ and $2i$.

As a result, the application of $M_1$, $M_2$, \dots,  $M_{n/2}$ on the initial input state ${\left| {0}^{n+t} \right\rangle }$ would produce the final output state $\left| x=s \right\rangle $ $\left|q=n-1 \right\rangle $, and thus the target string $s$ can be identified by measuring the $n$-qubit $x$ register in the computational basis. To explicitly explain the working principle of the above quantum algorithm, we reveal the overall effect of each  subroutine $M_i$ in Lemma \ref{Mi} in detail.

\begin{lemma}
\label{Mi}
Denote the integer $q^{(0)}=0$,  ${{q}^{(i)}}=2i-1(i=1,2,...,n/2)$, and ${{s}}^{(i)}={{s}_{1}}{{s}_{2}}..{{s}_{q^{(i)}-1}}$ as the $(2i-2)$-bit prefix of $s$ (${{s}}^{(1)}=\varnothing $ ). Then each subroutine $M_i$ ($i$=1,2,...,n/2) consisting of $H_{2i-1}, H_{2i}$, $Q_i$, $O_s$, and $R$ as shown in Figure~\ref{figQuanAL} has the effect
\begin{equation}\label{Mi:}
    M_i:{\left| {\psi _{0}} \right\rangle} =\left|x= {s}^{(i)}000^{n-2i} \right\rangle \left|q= {q}^{(i-1)} \right\rangle{\to} \left|x= {s}^{(i)}s_{2i-1}s_{2i}0^{n-2i} \right\rangle \left|q= {q}^{(i)} \right\rangle.
\end{equation}
\end{lemma}

\begin{proof}
For brevity, we first illustrate some key facts related to our quantum algorithm. We denote a set of four specific $n$-bit strings as $T=\{{{x}^{(k)}}={{s}}^{(i)}k{{0}^{n-2i}}:k={k_1}{k_2}\in {{\{0,1\}}^{2}}\}$, and according to the definitions of the functions in Eq.~\eqref{Ksx} and Eq.~\eqref{fsqx} we have
 \begin{equation}\label{}
   lcp(s,{{x}^{(k)}}) =
    \begin{cases}
     2i-2, & {k}_{1}\ne {s}_{2i-1} ;\\
       2i-1, & {{k}_{1}}={{s}_{2i-1}}, {{k}_{2}}\ne {{s}_{2i}} ;\\
       \ge 2i,  &{k_1}{k_2}={s_{2i-1}}{s_{2i}}
    \end{cases}
\end{equation}
by noting the $(2i-2)$-bit prefix of each $x^{(k)}$ is ${{s}}^{(i)}={{s}_{1}}{{s}_{2}}..{{s}_{2i-2}}$ and therefore
 \begin{equation}\label{fsxkqi}
    {{f}_{s}}({x}^{(k)},q^{(i)}=2i-1) =
    \begin{cases}
     0, & \text{$k\ne{{s}_{2i-1}}{{s}_{2i}}$  (three different such  $k$) } ;\\
       1, & \text{$k={{s}_{2i-1}}{{s}_{2i}}$  (only one such $k$) }
    \end{cases}
\end{equation}
for four ${{x}^{(k)}}$ in $T$. Besides, it can be directly verified that for any $k_0\in {\{0,1\}}^{2}$, the operator $R$  in Eq.~\eqref{Rmatrix}
can transform a 2-qubit superposed state into the basis state $\left| k_0 \right\rangle$   as
\begin{equation}\label{Rtransform}
  R: \sum\limits_{k\in {{\{0,1\}}^{2}}}{{\alpha}_k \left| k \right\rangle  }\to \left| k_0 \right\rangle
\end{equation}
when the coefficients are
\begin{equation}\label{alphak}
  \alpha_k=\begin{cases}
     \frac{1}{2}, & k\ne k_0 ;\\
       -\frac{1}{2}, & k=k_0 .
    \end{cases}
\end{equation}
Based on these facts and notations, the overall effect of $M_i$ in Eq.~\eqref{Mi:} is realized by following steps:

(1) Two Hadamard gates $H_{2i-1}H_{2i}$ and the operator $Q_i$ together transform the state ${\left| \psi_0 \right\rangle}$ on the left hand side of Eq.~\eqref{Mi:} into
\begin{equation}\label{psi1}
    \left| {\psi _{1}} \right\rangle =\frac{1}{2}\sum\limits_{k\in {{\{0,1\}}^{2}}}{\left| {s}^{(i)}{k}0^{n-2i} \right\rangle}\left| {q}^{(i)} \right\rangle \\=\frac{1}{2}\sum\limits_{k\in {{\{0,1\}}^{2}}}{\left| x^{(k)} \right\rangle}\left| {q}^{(i)} \right\rangle.
\end{equation}

(2) By utilizing Eq.~\eqref{fsxkqi}, it can be derived the quantum oracle operator $O_s$ defined in Eq.~\eqref{qxOs} can transform the state $\left| {{\psi }_{1}} \right\rangle $ of Eq.~\eqref{psi1} into
\begin{equation}\label{psi2}
  \left| {{\psi }_{2}} \right\rangle = \sum\limits_{k\in {{\{0,1\}}^{2}}}{{\alpha}_k \left| x^{(k)} \right\rangle  }\left| {q}^{(i)} \right\rangle = \left| {{s}^{(i)}}  \right\rangle (\sum\limits_{k\in {{\{0,1\}}^{2}}}{\alpha_k}\left| {k} \right\rangle  ) \left| {0^{n-2i}} \right\rangle
  \left| {q^{(i)}} \right\rangle
\end{equation}
with the coefficients
\begin{equation}
  \alpha_k=\begin{cases}
    \frac{1}{2}, & k\ne {s_{2i-1}}{s_{2i}} ;\\
       -\frac{1}{2}, & k={s_{2i-1}}{s_{2i}} .
    \end{cases}
\end{equation}

(3) In the end, the operator $R$ in 	Eq.~\eqref{Rmatrix} acting on qubits $2i-1$ and $2i$ can transform $\left| {\psi_2} \right\rangle$ of Eq.~\eqref{psi2} into
\begin{equation}\label{psi3}
    \left| {\psi_3} \right\rangle = \left| {{s}^{(i)}}  \right\rangle \left| {s_{2i-1}}{s_{2i}} \right\rangle   \left| {0^{n-2i}} \right\rangle
  \left| {q^{(i)}} \right\rangle
\end{equation}
by using 	Eq.~\eqref{Rtransform} and 	Eq.~\eqref{alphak}, which thus proves Eq.~\eqref{Mi:} of Lemma~\ref{Mi}.
\end{proof}

According to Lemma \ref{Mi}, the subroutines $\{{M_i:i=1,2,...,n/2}\}$ in Figure~\ref{figQuanAL} transforms the initial input state as:
\begin{align*}
    {\left| x={0}^{n} \right\rangle } {\left| q=0 \right\rangle }&{\xrightarrow{{M_1}}} {\left| {x={s_1}{s_2}{0^{n-2}}} \right\rangle } {\left| {q=1} \right\rangle }\\
   &{\xrightarrow{{M_2}}}{\left| {x={s_1}{s_2}{s_3}{s_4}{0^{n-4}}} \right\rangle } {\left| {q=3} \right\rangle }\\
   &{\xrightarrow{{M_3}}}\cdots {\xrightarrow{{M_{n/2}}}} {\left| {x=s} \right\rangle } {\left| {q=n-1} \right\rangle },
\end{align*}
and thus the secret string $s$ can be identified by measuring the $x$ register of the output state. Therefore, the whole quantum algorithm totally employs $n/2$ quantum oracle queries for identifying a secret string $s$, which outperforms  classical deterministic  algorithms that use at least $n$ queries in both the worst and average cases (see Theorem~\ref{theorem1}). This double speedup of our quantum algorithm comes from the intrinsic quantum parallelism, that is, the ability to evaluate four function values $ \{{{f}_{s}}({x}^{(k)},q^{(i)}):k\in {{\{0,1\}}^{2}}\}$ by each quantum oracle query $O_s$ (i.e. from  Eq.~\eqref{psi1} to Eq.~\eqref{psi2} ) and then to extract two bits of desired information $s_{2i-1}s_{2i}$ by interference via the operator $R$ (i.e. from  Eq.~\eqref{psi2} to Eq.~\eqref{psi3} )
in each subroutine $M_i$.

Finally, we address the case with  odd $n$. The idea is to  first obtain the $(n-1)$-bit prefix of the secret $s$ by $(n-1)/2$ quantum queries following the idea in Lemma~\ref{Mi} and then identify
the last bit of $s$ by a classical query. More precisely,
a sequence of $(n-1)/2$  subroutines $\{{M_i:i=1,2,...,(n-1)/2}\}$ constructed in the same way as those in Figure~\ref{figQuanAL} are applied to an input state  $\left|{0}^{n+t}  \right\rangle $ with $t=\left\lceil {\log _2}(n-1) \right\rceil $, and then the $x$ register of the output state is measured to obtain a string $x\in {{\{0,1\}}^{n}}$ revealing ${{s}_{1}}{{s}_{2}}...{{s}_{n-1}}={{x}_{1}}{{x}_{2}}...{{x}_{n-1}}$. Next, one classical query on $(x,q=n-1)$ is performed to identify the last bit of $s$ as
\begin{equation}\label{}
    s_n =
 \begin{cases}
 {{x}_{n}}\oplus 1, &  {{f}_{s}}(x,q=n-1)=0;\\
        {{x}_{n}}, & {{f}_{s}}(x,q=n-1)=1
    \end{cases}
    \label{odd_sn}
\end{equation}
according to  Eq.~\eqref{fsqx}. Therefore, we  employ $(n-1)/2$ quantum queries and 1 classical query together to identify $s$  for odd $n$. Note that there are totally $\left\lceil n/2\right\rceil$ queries. Thus, we have completed the proof of Theorem \ref{theorem2}.
\end{proof}
  In the next section, we demonstrate our quantum algorithm on an IBM cloud quantum computer  for several problem instances and evaluate the experimental performances.

\subsection{Experimental demonstrations on the IBM quantum computer}

For demonstrating our quantum  algorithm in the commonly used circuit model of quantum computation, we would employ a simple set of quantum gates consisting of the single-qubit Pauli $X$ and $Z$ gates, Hadamard gate $H$, $Z$-axis rotation gate ${{R}_{Z}}(\theta )$, square-root of $X$ gate $\sqrt{X}$, together with the two-qubit controlled-NOT (CNOT) gate. Unitary matrices and circuit symbols for these gates are listed in Table~\ref{tab:my_label}.

\begin{table}[htp]
\setlength{\belowcaptionskip}{10pt}
      \caption{Unitary matrices and circuit symbols for quantum gates used in this paper.}
\centering
    \begin{tabular}{ccc}
        \toprule
        \textbf{Quantum Gate}   & \textbf{Unitary Matrix}  & \textbf{Circuit Symbol} \\
        \midrule\\
        Pauli-$X$ & $ X:= \bigg( \begin{array}{cc}
     0 & 1 \\
     1 & 0
\end{array} \bigg)$
    & \Qcircuit  {& \gate{X} & \qw} \\  \\
   Pauli-$Z$ & $
    Z:=\bigg( \begin{array}{cc} 1 & 0 \\ 0 & -1 \end{array} \bigg)$ & \Qcircuit  {& \gate{Z} & \qw} \\
    \\ Hadamard & $
    H:=\frac{1}{\sqrt{2}}\bigg( \begin{array}{cc}1 & 1 \\ 1 & -1 \end{array} \bigg)$ & \Qcircuit  {& \gate{H} & \qw} \\
    \\  $Z$-axis rotation & $
    R_Z(\theta):=\bigg( \begin{array}{cc} e^{-i\theta/2} & 0 \\ 0 & e^{i\theta/2} \end{array} \bigg)$ & \includegraphics[height=0.8 cm]{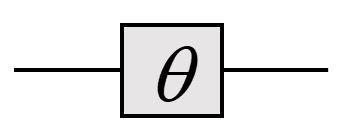} \\
   \\  square-root of $X$ & $
    \sqrt{X}:=\frac{1}{2}\bigg( \begin{array}{cc} 1+i & 1-i \\ 1-i & 1+i \end{array} \bigg)$ & \Qcircuit  {& \gate{\sqrt{X}} & \qw} \\

    \\

         controlled-NOT &
    $CNOT:=\begin{pmatrix} 1 & 0 & 0 & 0 \\ 0 & 1 & 0 & 0 \\ 0 & 0 & 0 & 1 \\ 0 & 0 & 1 & 0 \end{pmatrix}$ &
    \begin{tabular}{c}
          \Qcircuit  {
 & \ctrl{1} &  \qw \\
& \targ &  \qw
}
\end{tabular}\\ \\
        \bottomrule
    \end{tabular}
    \label{tab:my_label}
\end{table}

In order to experimentally demonstrate our quantum algorithm on real quantum hardware, e.g., IBM cloud superconducting quantum computers~\cite{IBM Quantum}, a series of key issues need to be taken into account:

(1) First, the unitary operators $Q_i$, $O_s$, and $R$ shown in Figure~\ref{figQuanAL} are supposed to be explicitly constructed  using basic gates, which is called quantum circuit synthesis. In particular, $Q_i$  can be built from a series of Pauli $X$ gates for converting a basis state $\left| {q}^{(i-1)} \right\rangle$ into $\left| {{q}}^{(i)} \right\rangle$, the quantum oracle $O_s$ defined in 	Eq.~\eqref{qxOs} can be expressed as

\begin{equation}\label{OsMatrix}
    O_s=\sum\limits_{(x,q)\in {\{0,1\}}^{n+t}} {{(-1)^{{f}_{s}(x,q)}}\left| x,q \right\rangle  \left\langle  x,q \right|},
\end{equation}
which is actually a diagonal unitary matrix with diagonal elements $\pm 1$ and can be realized by \{CNOT, $R_Z(\theta)$\} gate set \cite{Bullock2004,Welch2014},
and the operator $R$ in Eq.~\eqref{Rmatrix} can be decomposed into $R=(H\otimes I)(Z\otimes X)\text{CNOT}(H\otimes I)$ as shown in  Fig.~\ref{n=2circuit}(b).

\begin{figure}[H]
    \centering
    \includegraphics[width=0.9\textwidth]{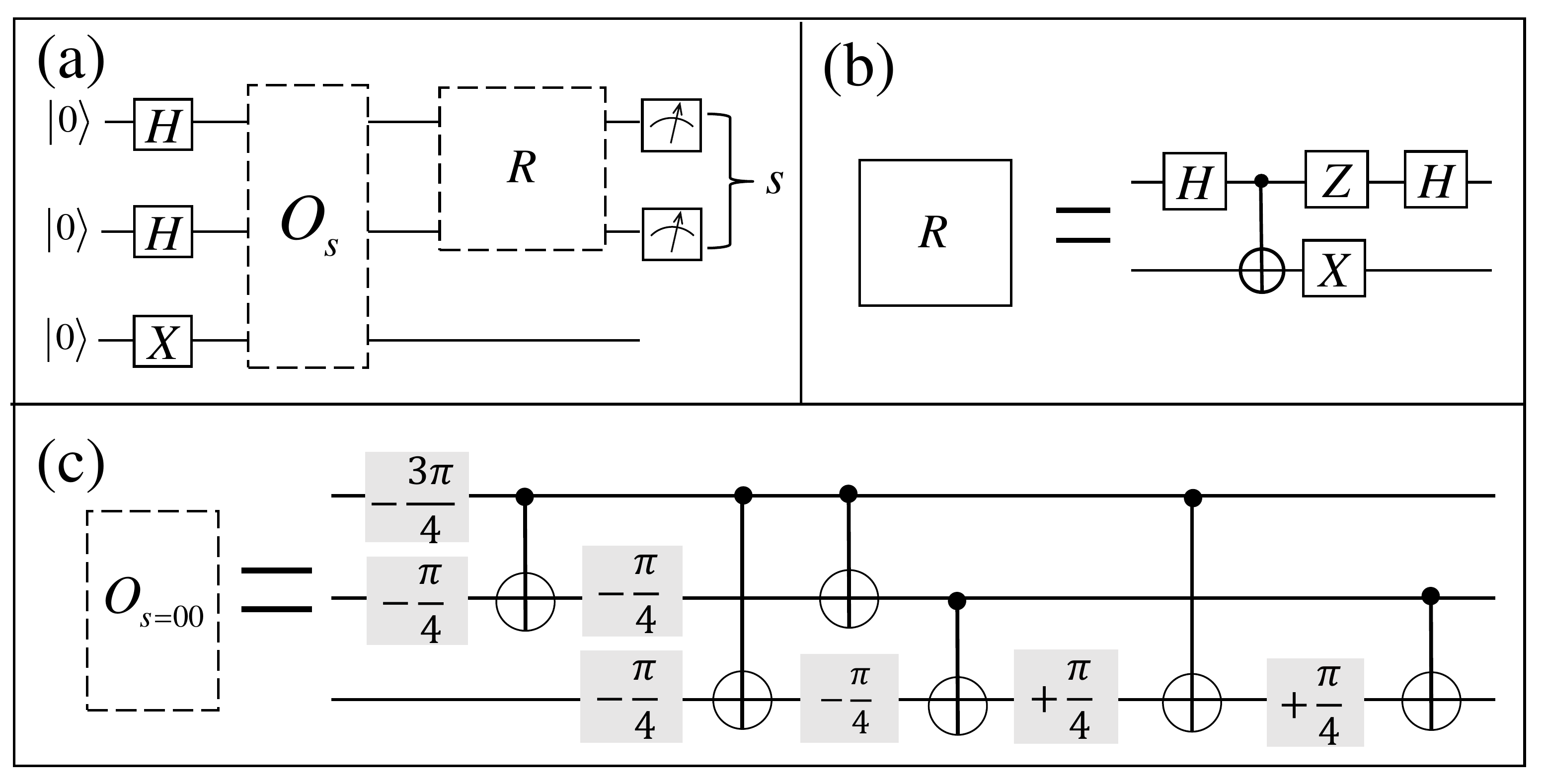}
    \caption{(a) Schematic of the circuit to demonstrate our quantum  algorithm for cases with  $n=2$, with the fixed operator $R$ of Eq.~\eqref{Rmatrix}  shown in (b) and the implementation of an oracle operator example $O_{s=00}$ = diag(-1,-1,-1,1,1,1,1,1) shown in (c). The details of the employed quantum gates are listed in  Table.~\ref{tab:my_label}. }
    \label{n=2circuit}
\end{figure}

(2) Second, considering both available physical gates and restricted qubit connectivity of employed quantum hardware, any synthesized circuit needs to be transpiled into a feasible circuit that can be executed in the practical hardware platform with acceptable outcomes. This essential processing procedure, including gate conversion and qubit mapping strategies, is usually called quantum circuit compilation \cite{Leymann2020,Kusyk2021} and has attracted a great deal of attention in the NISQ era \cite{Martinez2016,Itoko2019,Alam2020,Nash2020,Itoko2020,Bandic2022}. Here we implement customized circuit compilation  procedures aimed at running each synthesized circuit on the hardware named   $ibmq\_quito$ with its topology structure shown in Fig.~\ref{fig_compile}(a) and available gate set $G_{ibm}$ = \{CNOT, $R_Z(\theta)$, $\sqrt{X}$, $X$\}, including a Hadamard gate decomposition $H\simeq {R_Z(\pi/2)}{\sqrt{(X)}}{R_Z(\pi/2)}$ (differing only by an unimportant global phase factor) and a CNOT identity as shown in Fig.~\ref{fig_compile}(b) for implementing a CNOT between two unconnected qubits via four usable CNOT gates.

(3) Finally, we consider to take specific optimization techniques as shown in Fig.~\ref{fig_compile}(c) for further reducing the gate counts as well as circuit depth of compiled circuits, including:
\begin{equation}
\begin{aligned}
   R_Z\  {\rm gate\  merging:}\ &   R_Z{(\theta_2)}R_Z{(\theta_1)}=R_Z{(\theta_1+\theta_2)},\\ {\rm CNOT\ cancellation :}\ & {\rm CNOT} \cdot {\rm CNOT}=I,\\{\rm Commutation\ rules:}\ & {\rm CNOT} (R_Z{(\theta)} \otimes I)=(R_Z{(\theta)} \otimes I) {\rm CNOT},\\ &{\rm CNOT} (I\otimes R_Z{(\theta_2)}) {\rm CNOT} (I\otimes R_Z{(\theta_1)})\\=&
(I\otimes R_Z{(\theta_1)}) {\rm CNOT} (I\otimes R_Z{(\theta_2)}) {\rm CNOT}.
   \end{aligned}
   \label{CircuitOpti}
\end{equation}

\begin{figure}[H]
    \centering
    \includegraphics[width=0.7\textwidth]{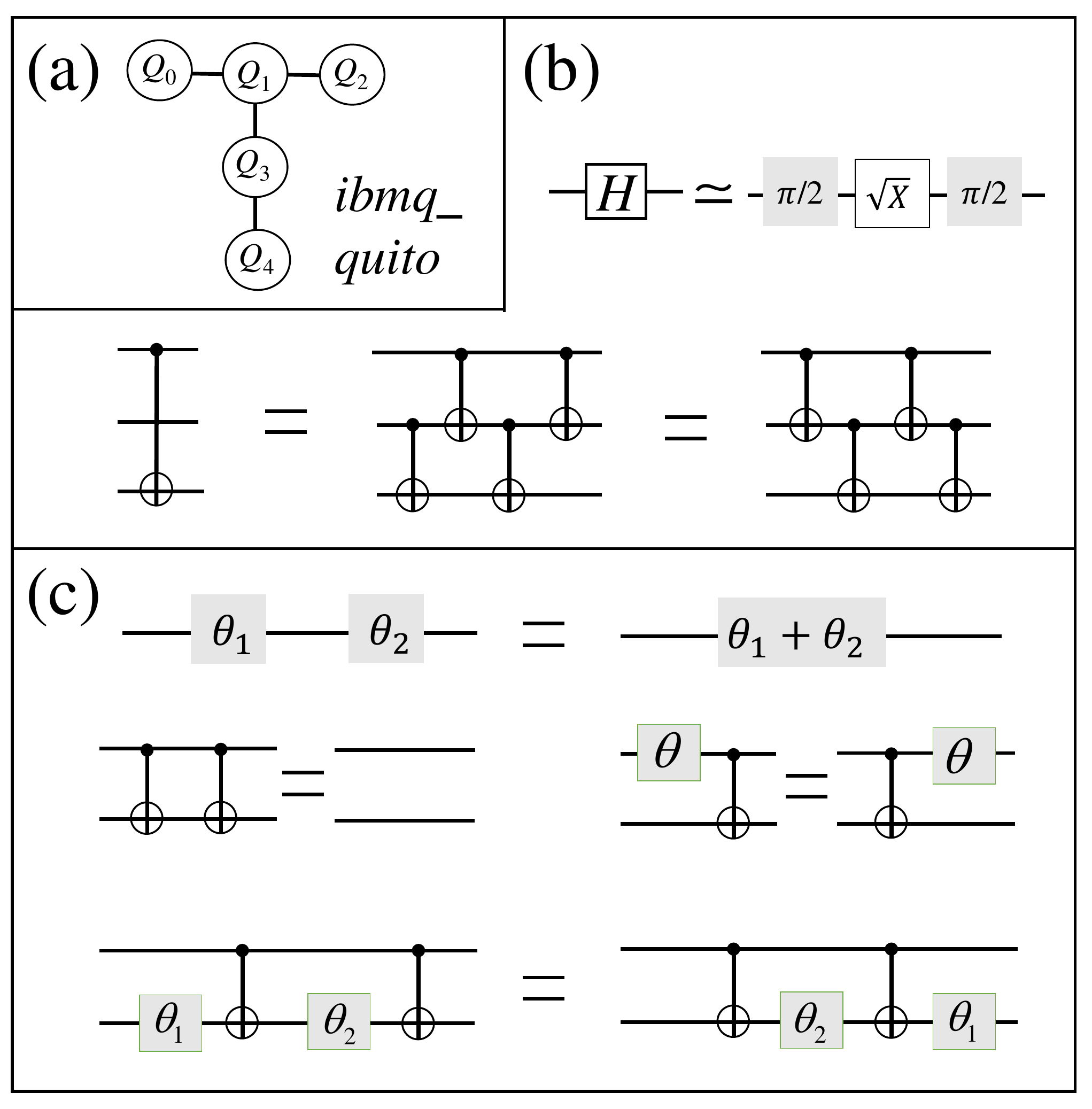}
    \caption{(a)Topology structure of the IBM quantum hardware $ibmq\_quito$. (b) Gate transformations used for quantum circuit compilation. (c) Some useful quantum circuit optimization techniques in Eq.~\eqref{CircuitOpti}. }
    \label{fig_compile}
\end{figure}

Based on the above process, we can bridge the gap between our quantum algorithm described at an abstract level and its practical implementation on real quantum hardware for any problem size in principle. In the following, we demonstrate the even case with $n=2$ and the odd case with $n=3$ on the IBM quantum hardware  $ibmq\_quito$.

For $n=2$, we use a three-qubit circuit with $n=2$ and $t=1$ for performing our quantum algorithm as shown in Fig.~\ref{n=2circuit}(a),  where the operator $Q_1$ acting on the $q$ register is a $X$ gate and each problem instance $s$ determines a specific oracle operator $O_s$ in Eq.~\eqref{OsMatrix}. For example, in Fig.~\ref{n=2circuit}(c) we present a construction of an oracle example $O_{s=00}$ = diag(-1,-1,-1,1,1,1,1,1) using six CNOT gates and seven ${R_Z}(\theta)$ gates. Next, we compile this circuit into a new one that fits the ${Q_0}-{Q_1}-{Q_2}$ linear structure and gate set $G_{ibm}$ of $ibmq\_quito$ using the equivalent gate transformations in Fig.~\ref{fig_compile}(b)
followed by the application of optimization techniques in Fig.~\ref{fig_compile}(c), which would lead to a final transpiled circuit consisting of 9 CNOT, 10 $R_Z(\theta)$, 4 $\sqrt{X}$, and 2 $X$ gates
with a circuit depth 15 as shown in Fig.~\ref{n=2circuitdata}(a). Besides, the oracle constructions and transpiled circuits for other three instances $s=01$, 10, and 11 have similar forms as provided in Supplementary Information \ref{suppA}.

\begin{figure}[H]
    \centering
    \includegraphics[width=0.95\textwidth]{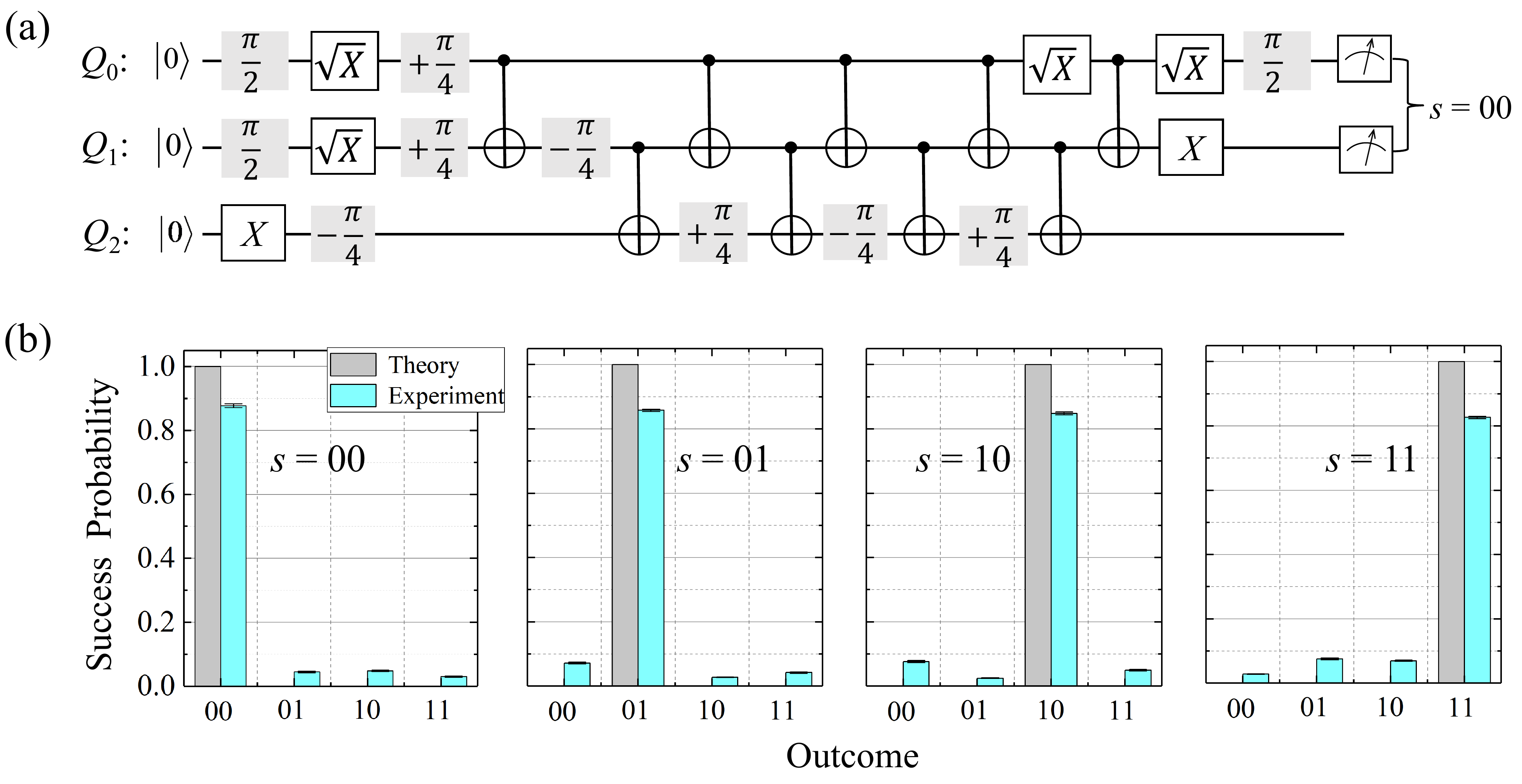}
    \caption{(a) Transpiled circuit acting on three physical qubits \{$Q_0,Q_1,Q_2$\} of $ibmq\_quito$ for identifying secret $s=00$ as an example, which is obtained from the original circuit in Fig.~\ref{n=2circuit} through compilation and optimization procedures in Fig.~\ref{fig_compile}. (b) Experimental results for demonstrating all instances $s\in {\{0,1\}}^2$ via their associated transpiled circuits. The success probabilities are $87.7(6)    \%$, $85.8(3)    \%$, $84.8(4) \%$, and $82.6(3) \%$  for  $s=00$, 01, 10, and 11, respectively,  with 5 trials and the sample size 8192 ($5\times 8192$ realizations for each $s$). }
    \label{n=2circuitdata}
\end{figure}

All four instances for $n=2$ are tested in our experiment, where we employ a figure of merit called the algorithm success probability (ASP) to denote the probability of correctly identifying the target string as the experiment outcome.
For demonstrating an instance $s\in \{0,1\}^{2}$ on the device, we experimentally implement 5 trials and in each trial we consider 8192 repetitions of the transpiled circuit to record outcome data for calculating the ASP. At the time the device $ibmq\_quito$ was accessed, the average CNOT error was $1.080\times 10^{-2}$ and average readout error was $4.424\times 10^{-2}$, average $T_1=98.55 \mu s$ (decay time from the excited state to ground state), average  $T_2=101.74 \mu s$ (decoherence time). As the distributions over possible measurement outcomes show in Fig.~\ref{n=2circuitdata}(b),
the  success probabilities for identifying secret $s=00$, 01, 10, and 11 in the experiment are $87.7(6)    \%$, $85.8(3)    \%$, $84.8(4) \%$, and $82.6(3) \%$, respectively. As an overall evaluation, the average success probability defined over all problem instances is $85.3\%$ in the above experimental demonstration.

\begin{figure}[H]
    \centering
    \includegraphics[width=0.9\textwidth]{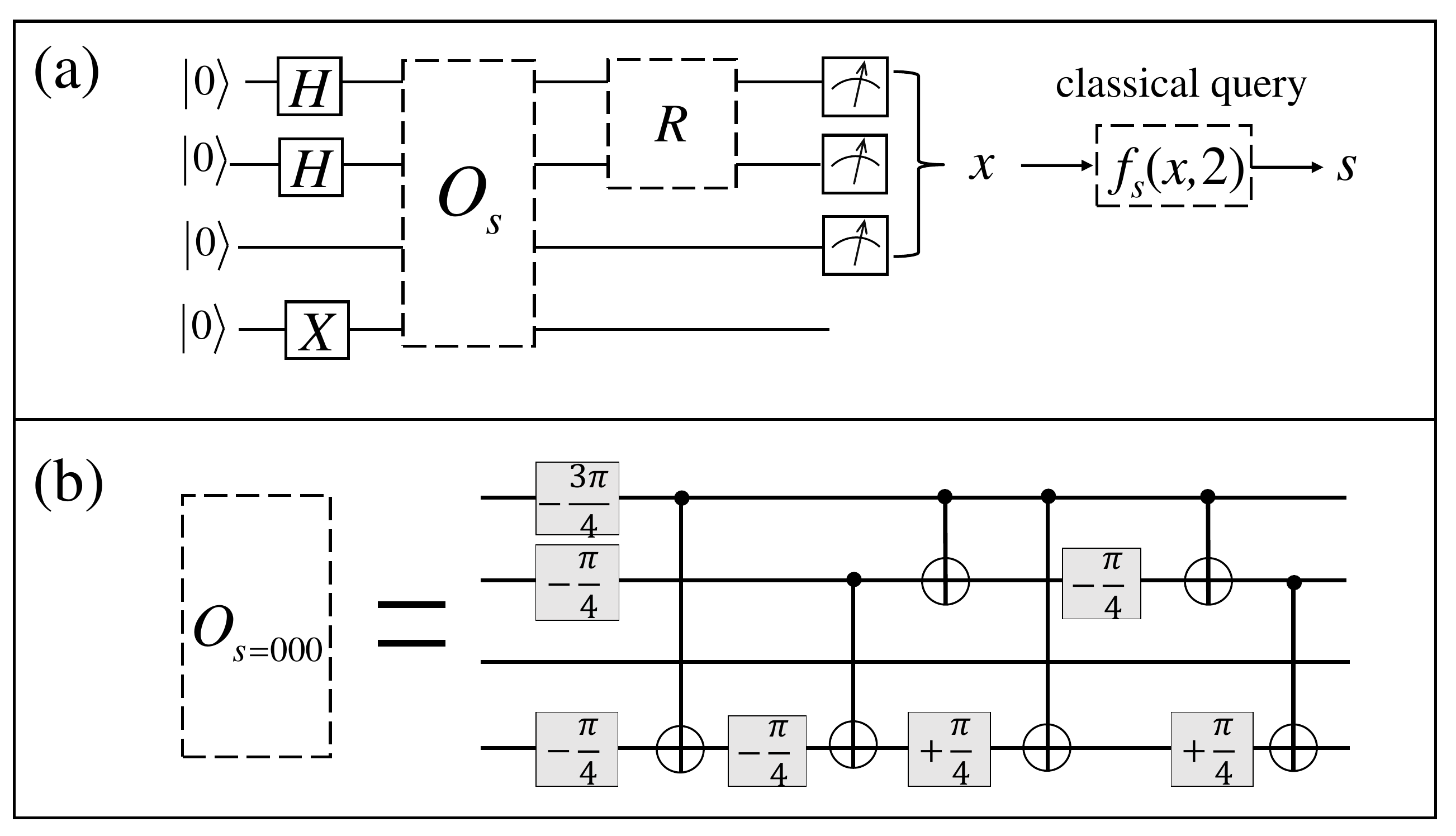}
    \caption{(a) Schematic of the quantum circuit followed by a classical query to demonstrate our quantum  algorithm for cases with  $n=3$,  with the fixed operator $R$ shown in Fig.~\ref{n=2circuit}(b) and implementation of an oracle operator example $O_{s=000}$ = diag(-1,-1,-1,-1,-1,1,-1,1,1,1,1,1,1,1,1,1) shown in (b). }
    \label{n=3circuit}
\end{figure}

For $n=3$, we use a four-qubit circuit with $n=3$ and $t=1$ followed by a classical query for performing our quantum algorithm as shown in Fig.~\ref{n=3circuit}(a). For example, in Fig.~\ref{n=3circuit}(b) we present a construction of an oracle example $O_{s=000}$= diag(-1,-1,-1,-1,-1,1,-1,1,1,1,1,1,1,1,1,1) using six CNOT gates and seven ${R_Z}(\theta)$ gates.  Next, we compile this circuit by assigning four qubits of Fig.~\ref{n=3circuit}(a) to physical qubits $\{{Q_0},{Q_1},{Q_2},{Q_3}\}$ in the T-shape structure of $ibmq\_quito$ in Fig.~\ref{fig_compile}(a)
and then using circuit transformation techniques in Figs.~\ref{fig_compile}(b) and  \ref{fig_compile}(c), which would lead to a final transpiled circuit consisting of 9 CNOT, 10 $R_Z(\theta)$, 4 $\sqrt{X}$, and 2 $X$ gates  with a circuit depth 15 as shown in Fig.~\ref{n=3circuitdata}(a). Also, the oracle constructions and transpiled circuits for other instances $s$ have similar forms as provided in Supplementary Information \ref{suppB}. For completeness, we experimentally perform our algorithm  by taking 5$\times$8192 repetitions for each one of the eight instances $s\in \{0,1\}^3$ , and the success probabilities for identifying each instance are $85.3(7)\%$, $85.3(4)\%$, $83.0(9)\%$, $82.2(6)\%$, $82.3(5)\%$, $82.4(1)\%$, $79.7(4)\%$ and $79.5(5)\%$ as recorded in Fig.~\ref{n=3circuitdata}(b), respectively.
As an overall evaluation, the average success probability defined over all problem instances is $82.5\%$ in our experiment.

\begin{figure}[H]
    \centering
    \includegraphics[width=0.95\textwidth]{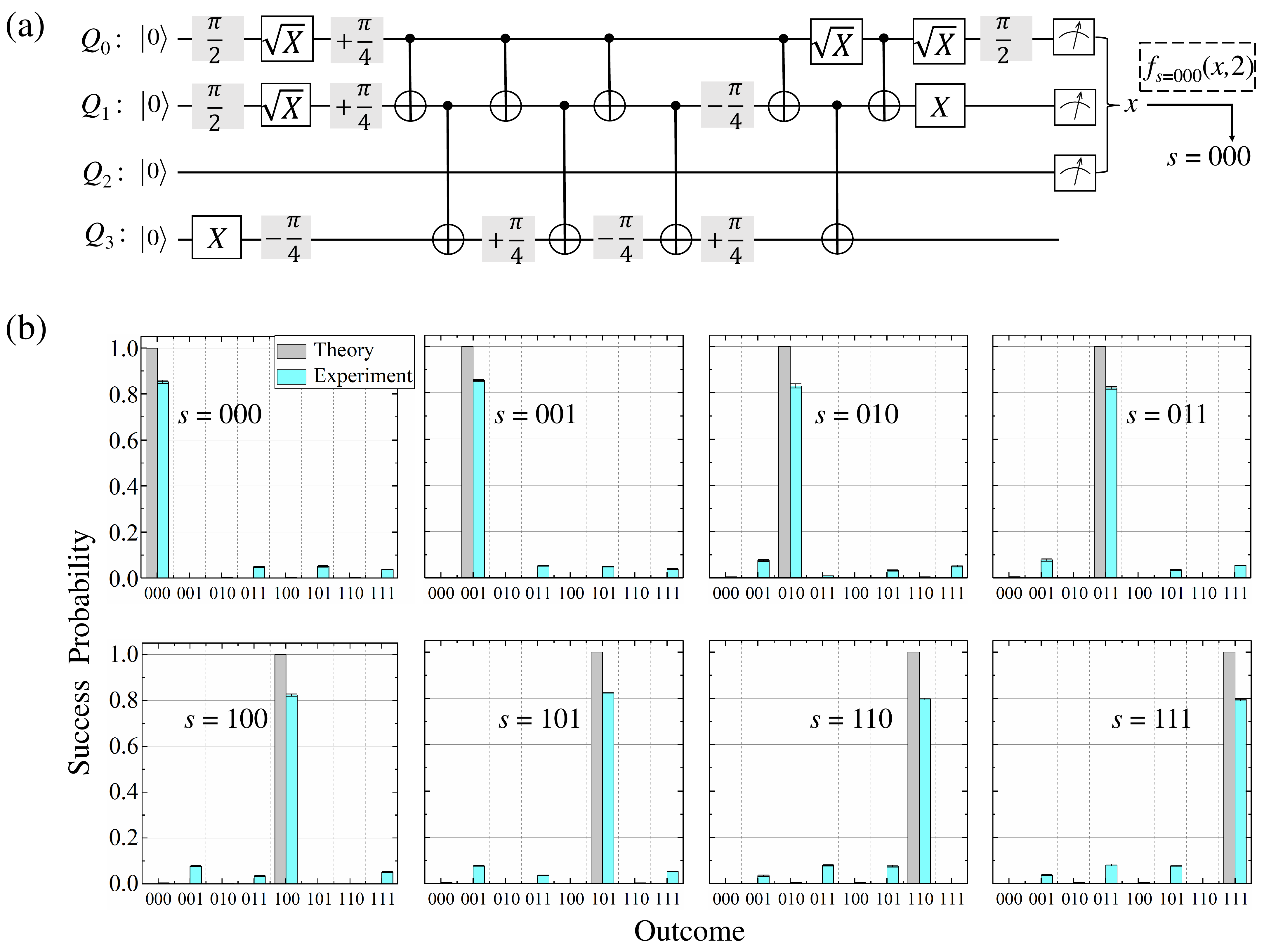}
    \caption{(a) Transpiled circuit acting on four physical qubits \{$Q_0,Q_1,Q_2,Q_3$\} of $ibmq\_quito$ followed by a classical query $f_{s=000}(x,2)$ for identifying secret $s=000$ as an example, which is obtained from the original circuit in Fig.~\ref{n=3circuit} through compilation and optimization procedures in Fig.~\ref{fig_compile}. (b) Experimental results for demonstrating all eight instances  $s\in {\{0,1\}}^3$ via their associated transpiled circuits together with a classical query, where the corresponding success probabilities are $85.3(7)\%$, $85.3(4)\%$, $83.0(9)\%$, $82.2(6)\%$, $82.3(5)\%$, $82.4(1)\%$, $79.7(4)\%$ and $79.5(5)\%$  , respectively,  with 5 trials and the sample size 8192  for each $s$.  }
    \label{n=3circuitdata}
\end{figure}

\section{Discussion}
\label{Discussion}

In this study, we discussed how to learn a secret  string $s\in {{\{0,1\}}^{n}}$ by querying the oracle  $f_{s}(x,q)$ that indicates  whether the length of the longest common prefix  of $s$ and $x$ is greater than $q$ or not. We first proved  that the  classical  query complexity of this problem is $n$ in both the worst and average case as well  as gave an optimal classical deterministic algorithm. Then, we proposed an exact quantum algorithm with $\left\lceil n/2\right\rceil$ query complexity for solving any problem instance, which thus provided a double speedup over its classical counterparts. Finally, we experimentally demonstrated our quantum algorithm on an IBM cloud quantum device by utilizing specific circuit design and compilation techniques, where the average success probability for all instances with $n=2$ or $n=3$ reached $85.3\%$ or $82.5\%$, respectively. Also, experiments on our quantum algorithm for larger-scale secret string   problems are likely to be performed by following a similar way. Besides, the $q$ register used in our quantum
algorithm can be alternatively represented by higher-dimensional qu$d$it $(d>2)$ systems for experimental demonstrations, e.g., using quantum photonic \cite{Erhard2020} or nuclear magnetic resonance (NMR) systems \cite{Gedik2015}. Since it is not yet clear whether our quantum algorithm is optimal, in future work we would consider an interesting open problem about the lower bound of exact or bounded-error quantum query complexity for this secret-string-learning problem.

\section*{Data Availability}
All the data supporting  the results of this work are available from the
 authors upon reasonable request.

\section*{Code Availability}The codes for designing circuits in this work are available from the authors  upon reasonable request.

\section*{Acknowledgements}
This work was supported by the National Natural Science Foundation of China (Grant Nos. 62102464, 61772565),
 the Guangdong Basic and Applied Basic Research Foundation
(Grant No. 2020B1515020050), the Key Research and Development
project of Guangdong Province (Grant No. 2018B030325001), and Project funded by
China Postdoctoral Science Foundation (Grant Nos. 2020M683049,
2021T140761).

\section*{Author Contributions}
L. Li conceived this problem. Y. Xu and S. Zhang  designed the quantum algorithm together. Y. Xu put forward the classical algorithm and proved its optimality, S. Zhang performed the quantum experiments and data processing.
L. Li and S. Zhang supervised the project, and all authors wrote the manuscript. Y. Xu and S. Zhang contributed equally to this work.

\section*{Competing Interests}
The authors declare no competing interests.

\newpage

\appendices
\setcounter{table}{0}
\setcounter{figure}{0}

\renewcommand{\thetable}{S\arabic{table}}
\renewcommand*{\theHtable}{\thetable}
\renewcommand{\thefigure}{S\arabic{figure}}
\renewcommand*{\theHfigure}{\thefigure}

\begin{center}
    \section*{Supplementary Information}
\end{center}

As mentioned in the main text, we experimentally demonstrate our quantum algorithm for all cases with $n=2$ and $n=3$ on the IBM quantum device named $ibmq\_quito$, and the calibration parameters for the day the device was accessed are shown in Table~\ref{quitoparam}. For completeness, in the following we give detailed descriptions of all employed quantum circuits.

\begin{table}[H]
\setlength{\belowcaptionskip}{10pt}
\caption{Calibration parameters of $ibmq\_quito$ from the IBM quantum website. In the CNOT error column, $i\_j$ indicates that the physical qubits $i$ and $j$ are the control and target for the CNOT gate, respectively.}
\resizebox{\linewidth}{!}{
\begin{tabular}{|c|c|c|c|c|c|c|c|c|c|}
\hline
Qubit & T1 (us) & T2 (us) & Freq. (GHz) & Anharm. (GHz) & Readout error & ID error & $\sqrt{x}$ (sx) error & Pauli-X error & CNOT error                   \\ \hline
Q0    & 79.19      & 126.78 &  5.301  & -0.33148     &$3.81 \times 10^{-2}$ & $3.23 \times 10^{-4}$  &$ 3.23 \times 10^{-4}$ & $3.23\times 10^{-4}$    & $0\_1:7.401\times 10^{-3}$             \\ \hline
Q1    &  117.96    &  132.4 &   5.081 & -0.31925   & $4.11\times 10^{-2}$   & $2.90 \times 10^{-4}$ &$2.90  \times 10^{-4}$ & $2.90\times 10^{-4}$      &\begin{tabular}{c} $1\_3:1.044 \times 10^{-2};$ \\$ 1\_2:6.435 \times 10^{-3};$\\ $1\_0:7.401 \times 10^{-3} $\end{tabular}  \\ \hline
Q2    &  95.79     &   115.86 & 5.322 &  -0.33232   & $7.17\times 10^{-2}$  & $2.74 \times 10^{-4}$ &$2.74  \times 10^{-4}$     &$2.74\times 10^{-4}$      & $2\_1:6.435 \times 10^{-3}$ \\ \hline
Q3    & 107.55   &  22.83 &   5.164   &  -0.33508   &  $3.41\times 10^{-2}$ & $3.44 \times 10^{-4}$ &$3.44 \times 10^{-4}$     &  $3.44\times 10^{-4}$    & \begin{tabular}{c} $3\_4:1.890 \times 10^{-2};$ \\$ 3\_1:1.044 \times 10^{-2}$\end{tabular}  \\ \hline
Q4    &  92.27   & 110.84  &    5.052 &  -0.31926   & $3.62\times 10^{-2}$  & $4.57 \times 10^{-4}$& $4.57 \times 10^{-4}$     & $4.57\times 10^{-4}$    &    $ 4\_3:1.890 \times 10^{-2}$         \\ \hline
\end{tabular}}\label{quitoparam}
\end{table}

\section{Details of quantum circuits for demonstrating cases with $n=2$}\label{suppA}

As introduced in the main text, the  designed quantum oracle and the whole transpiled circuit for $s=00$ have been  presented in Fig.~\ref{n=2circuit}(c) and Fig.~\ref{n=2circuitdata}(a) of main text, respectively. Similarly, the oracle constructions of $O_s$ in Eq.~\eqref{OsMatrix} of the main text and corresponding transpiled circuits for demonstrating other three cases $s$=01, 10, and 11 are explicitly shown in Fig.~\ref{figS1} and Fig.~\ref{figS2}, respectively. The experimental results for testing these cases are plotted in Fig.~\ref{n=2circuitdata} (b) of the main text.

\begin{figure}[H]
    \centering
    \includegraphics[width=0.8\textwidth]{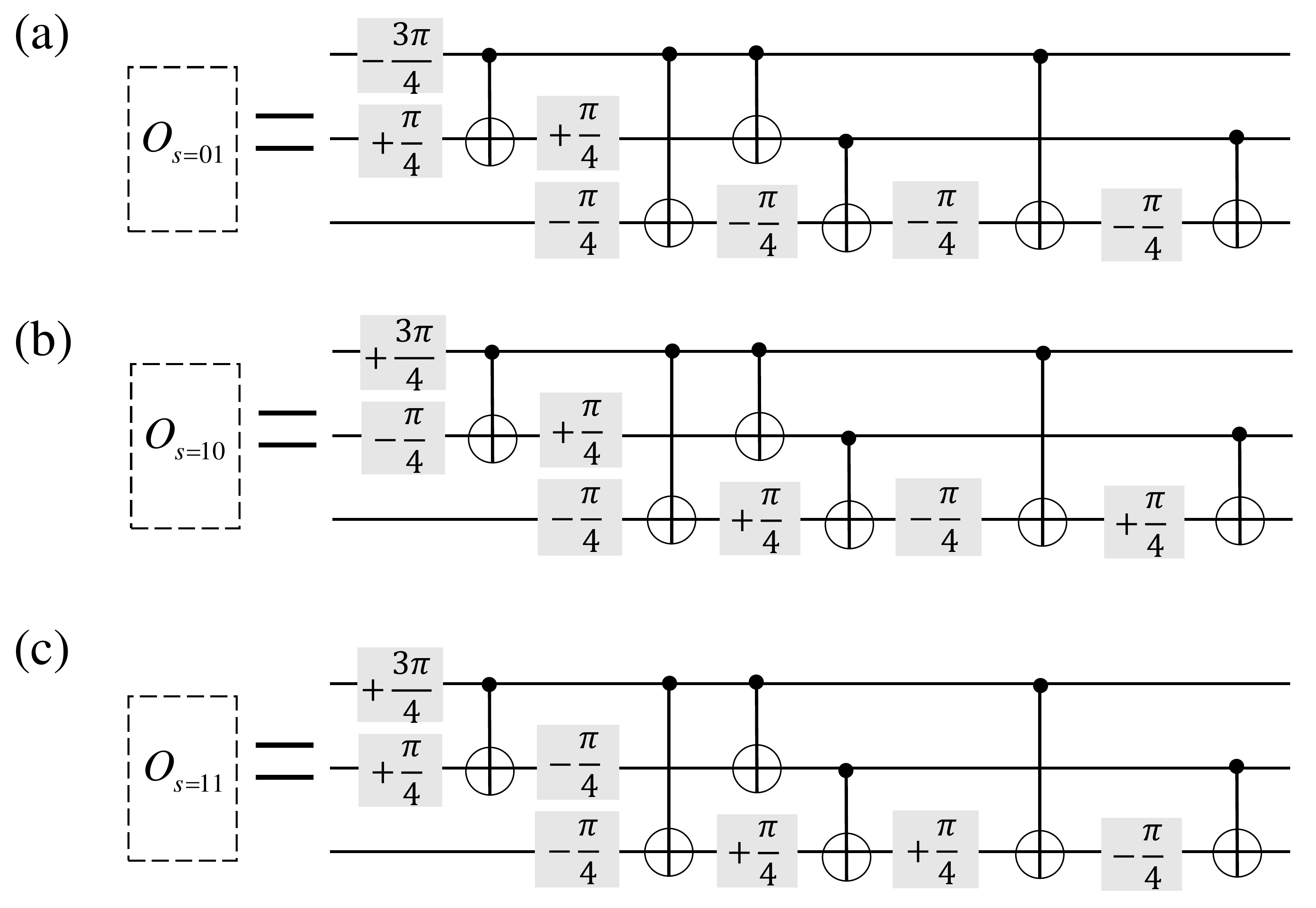}
    \caption{Quantum oracle constructions of (a)  $O_{s=01}$ = diag(-1,1,-1,-1,1,1,1,1), (b)  $O_{s=10}$ = diag(1,1,1,1,-1,-1,-1,1),  and (c)  $O_{s=11}$ = diag(1,1,1,1,-1,1,-1,-1), respectively. }
    \label{figS1}
\end{figure}

\begin{figure}[H]
    \centering
    \includegraphics[width=0.9\textwidth]{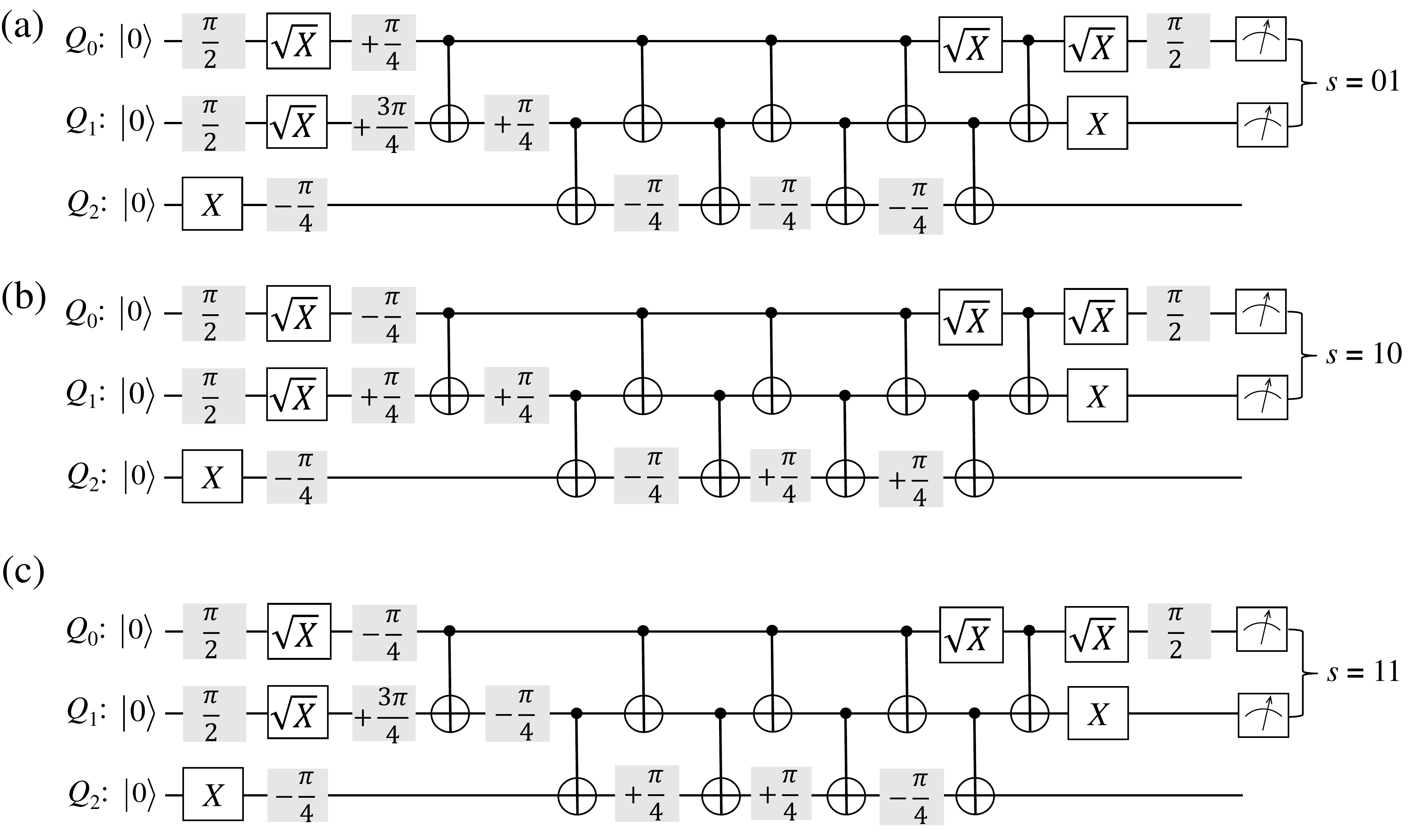}
    \caption{Overall transpiled circuits  for demonstrating (a) $s=01$, (b) $s=10$, and (c) $s=11$ on $ibmq\_quito$, respectively. }
    \label{figS2}
\end{figure}

\section{Details of quantum circuits for demonstrating cases with $n=3$}\label{suppB}

For demonstrating cases with $n=3$, two features of our algorithm are noteworthy. First, it can be verified that the quantum oracles are the same for any two instances ${s_1}{s_2}0$ and ${s_1}{s_2}1$ as shown in Fig.~\ref{figS3}.  Therefore, we only need to design and compile four different quantum circuits in total, that is, Fig.~\ref{n=3circuitdata}(a) in the main text  as well as Fig.~\ref{figS4} below to identify two bits ${s_1}{s_2}\in \{0,1\}^2$, supplemented with a final classical query $f_s(x,2)$ to identify $s_3$ of secret $s$. For example, the transpiled quantum circuit inside Fig.~\ref{n=3circuitdata}(a) of the main text together with a query ${f_{s=001}(x,2)}$ on the measured outcome string $x$ can be used for identifying secret $s=001$. Second, according to Eq.~\eqref{odd_sn} of the main text, our algorithm can correctly identify the secret string $s$ when the outcome measured from the quantum circuit in imperfect experiments is  ${s_1}{s_2}{s_3}$ or ${s_1}{s_2}{(s_3\oplus{1}})$, leading to the experimental success probabilities  recorded in Fig.~\ref{n=3circuitdata}(b) of main text for identifying each $s$.

\begin{figure}[H]
    \centering
    \includegraphics[width=0.8\textwidth]{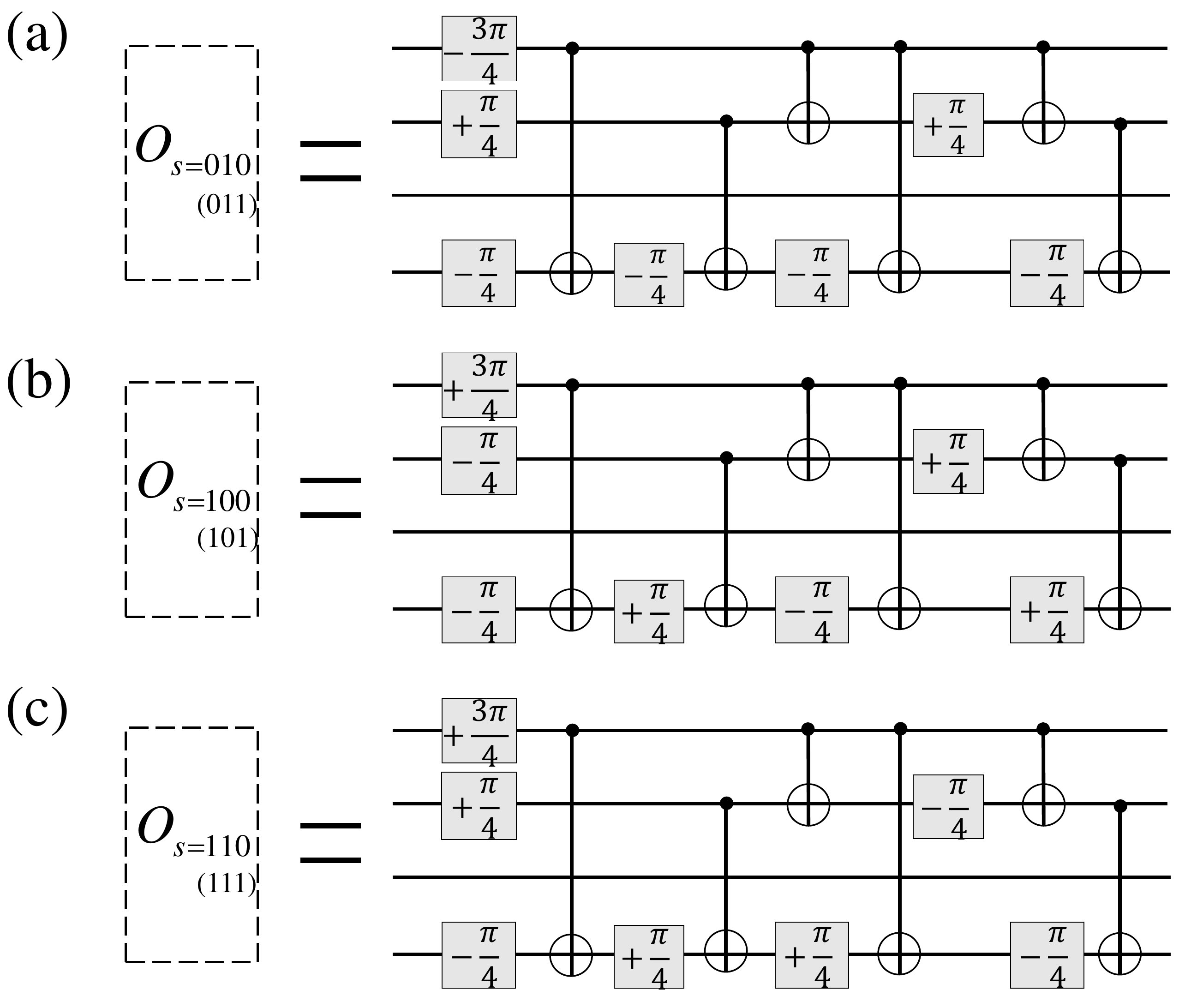}
    \caption{Quantum oracle constructions of (a) $O_{s=010(011)}$= diag(-1,+1,-1,+1,-1,-1,-1,-1,+1,+1,+1,+1,+1,+1,+1,+1 ), (b) $O_{s=100(101)}$= diag(+1, +1, +1, +1,  +1, +1, +1, +1,   -1,-1,-1,-1,-1,+1,-1,+1), and  (c) $O_{s=110(111)}$= diag(+1, +1, +1, +1,  +1, +1, +1, +1,   -1,+1,-1,+1,-1,-1,-1,-1), respectively. }
    \label{figS3}
\end{figure}

\begin{figure}[H]
    \centering
    \includegraphics[width=0.9\textwidth]{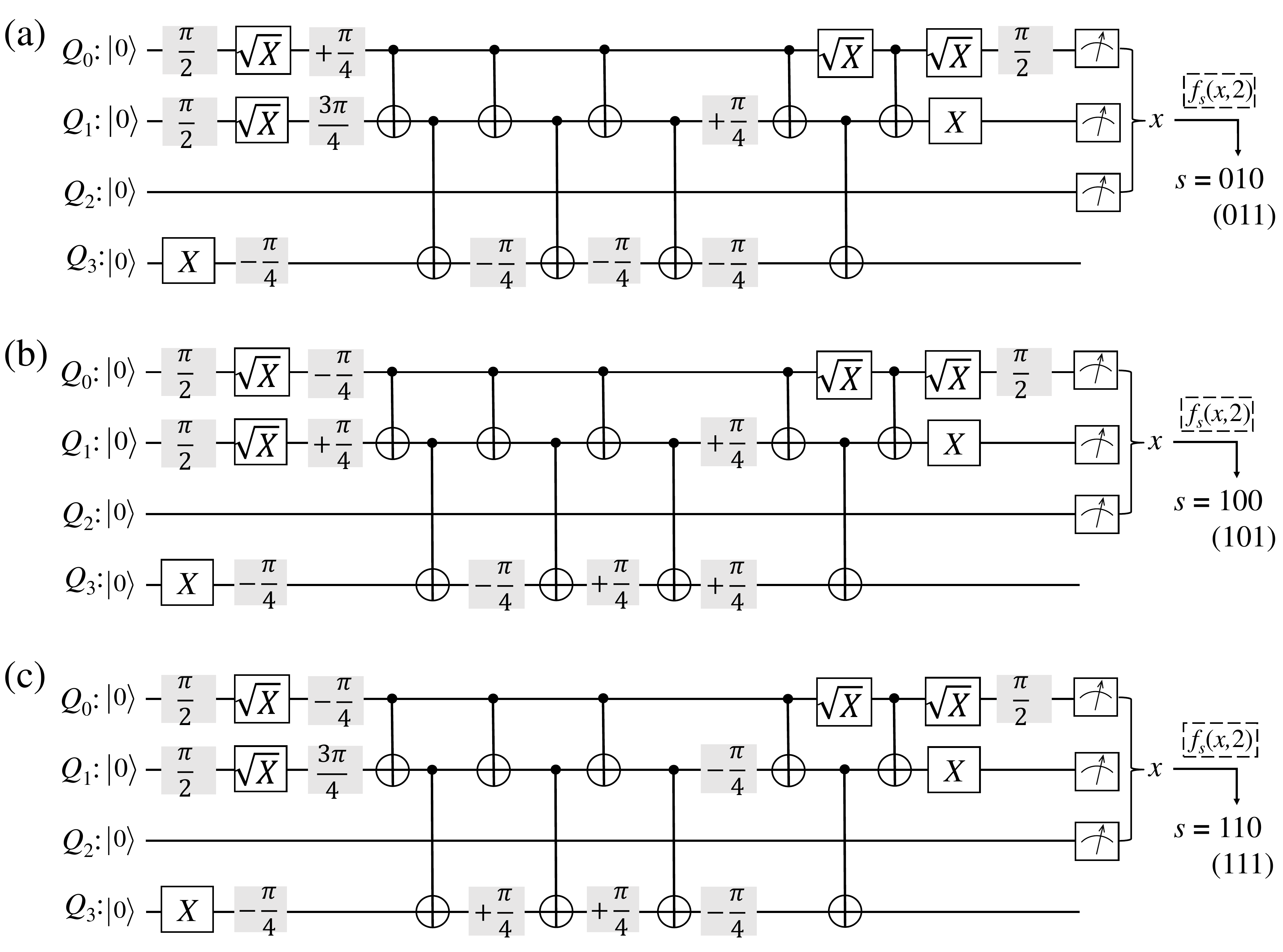}
    \caption{Overall transpiled quantum circuit acting on four physical qubits $\{Q_0,Q_1,Q_2,Q_3\}$ of $ibmq\_quito$ followed by a classical query $f_s(x,2)$ to identify (a) $s$=010 or 011, (b) $s$=100 or 101, and (c) $s$=110 or 111. }
    \label{figS4}
\end{figure}

\end{document}